\documentclass{article}

\usepackage{PRIMEarxiv}

\usepackage{times}
\usepackage{helvet}
\usepackage{courier}
\usepackage{url}
\usepackage{caption}
\usepackage{graphicx} 
\usepackage{hyperref}
\usepackage{comment}

\usepackage[ruled,vlined,linesnumbered]{algorithm2e}
\usepackage{comment}

\usepackage{environ}
\newif\ifshort
\ifshort
	\makeatletter
	\newenvironment{claimO}[1][]{
		\if\relax\detokenize{#1}\relax
		\expandafter\@firstoftwo
		\else
		\expandafter\@secondoftwo
		\fi
		{\begin{claim}[$\star$]}{\begin{claim}[#1, $\star$]}
			}
			{
			\end{claim}
		}
		\makeatother
		
	\makeatletter
	\newenvironment{lemmaO}[1][]{
		\if\relax\detokenize{#1}\relax
		\expandafter\@firstoftwo
		\else
		\expandafter\@secondoftwo
		\fi
		{\begin{lemma}[$\star$]}{\begin{lemma}[#1, $\star$]}
			}
			{
			\end{lemma}
		}
		\makeatother
		
		\makeatletter
		\newenvironment{propositionO}[1][]{
			\if\relax\detokenize{#1}\relax
			\expandafter\@firstoftwo
			\else
			\expandafter\@secondoftwo
			\fi
			{\begin{proposition}[$\star$]}{\begin{proposition}[#1, $\star$]}
				}
				{
				\end{proposition}
			}
			\makeatother
			
			\makeatletter
			\newenvironment{observationO}[1][]{
				\if\relax\detokenize{#1}\relax
				\expandafter\@firstoftwo
				\else
				\expandafter\@secondoftwo
				\fi
				{\begin{observation}[$\star$]}{\begin{observation}[#1, $\star$]}
					}
					{
					\end{observation}
				}
				\makeatother			
				
				\makeatletter
				\newenvironment{theoremO}[1][]{
					\if\relax\detokenize{#1}\relax
					\expandafter\@firstoftwo
					\else
					\expandafter\@secondoftwo
					\fi
					{\begin{theorem}[$\star$]}{\begin{theorem}[#1, $\star$]}
						}
						{
						\end{theorem}
					}
					\makeatother
					\makeatletter
					\newenvironment{corollaryO}[1][]{
						\if\relax\detokenize{#1}\relax
						\expandafter\@firstoftwo
						\else
						\expandafter\@secondoftwo
						\fi
						{\begin{corollary}[$\star$]}{\begin{corollary}[#1, $\star$]}
							}
							{
							\end{corollary}
						}
						\makeatother
						\NewEnviron{proofO}{}
						\NewEnviron{claimproofO}{}
\else

						
						

\fi

\usepackage{newfloat}
\usepackage{listings}
\usepackage{amsthm}
\usepackage{amssymb}
\usepackage{mathtools}
\usepackage{diagbox}
\usepackage{multirow}
\usepackage{graphicx,color}
\usepackage{algorithmic}
\usepackage[most]{tcolorbox}

\usepackage{enumitem}

\usepackage{xspace}
\usepackage{cleveref}
\usepackage{todonotes}
\usepackage{thm-restate}

\newtheorem{theorem}{Theorem}
\newtheorem{lemma}{Lemma}
\newtheorem{claim}{Claim}
\newtheorem{corollary}{Corollary}
\newtheorem{proposition}{Proposition}
\newtheorem{observation}{Observation}
\theoremstyle{definition}
\newtheorem{definition}{Definition}
\theoremstyle{definition}

\newcommand{\cqed}{\renewcommand{\qedsymbol}{$\lrcorner$}\qed}
\newenvironment{claimproof}{\noindent \emph{Proof of Claim~\theclaim.}}{\hfill\cqed\medskip}

\DeclareCaptionStyle{ruled}{labelfont=normalfont,labelsep=colon,strut=off} % DO NOT CHANGE THIS
\usepackage{amsmath}
\usepackage{amsthm}
\usepackage{amsfonts} 

\usepackage{xargs}

\usepackage[mathscr]{euscript}
\usepackage{xspace}

\newcommand{\Oh}{\mathcal{O}}

\newcommand{\OPT}{\textsc{OPT}}

\newcommand{\dist}{\mathrm{dist}}

\newcommand{\tw}{\operatorname{tw}}
\newcommand{\td}{\operatorname{td}}
\newcommand{\fvs}{\operatorname{fvs}}
\newcommand{\vc}{\operatorname{vc}}

\newcommand{\tpath}{$T$-path\xspace}
\newcommand{\tpathdel}{\textsc{\tpath-Deletion}\xspace}
\newcommand{\tpathedit}{\textsc{\tpath-Editing}\xspace}

\newcommand{\tpathadd}{\textsc{\tpath-Addition}\xspace}
\DeclareMathOperator{\operatorClassNP}{{\sf NP}}
\newcommand{\classNP}{\ensuremath{\operatorClassNP}}

\DeclareMathOperator{\operatorClassFPT}{{\sf FPT}\xspace}
\newcommand{\classFPT}{\ensuremath{\operatorClassFPT}\xspace}
\DeclareMathOperator{\operatorClassW}{{\sf W}}
\newcommand{\classW}[1]{\ensuremath{\operatorClassW[#1]}}
\DeclareMathOperator{\operatorClassParaNP}{{\sf Para-NP}\xspace}
\newcommand{\classParaNP}{\ensuremath{\operatorClassParaNP}\xspace}
\DeclareMathOperator{\operatorClassXP}{{\sf X}P\xspace}
\newcommand{\classXP}{\ensuremath{\operatorClassXP}\xspace}

\newtcolorbox{problem_box}[1]{enhanced,
  attach boxed title to top left={yshift=-3mm,yshifttext=-1mm, xshift=10mm},
  colback=gray!20!white,colframe=gray!80!black,
  boxrule=1pt,
  title={#1},
  coltitle=black,
  boxed title style={size=small,colframe=gray!80!black, colback=gray!40!white, boxrule=1pt},  arc = 2mm}
\newcommandx{\defsimpleproblem}[3][]{
    \begin{problem_box}{#1}
        {\textbf{Input:}} {#2}  \\
        {\textbf{Task:}} {#3}
    \end{problem_box}
}

\title{Structural Approach to Guiding a Present-Biased Agent}
\author{
  Tatiana Belova\\
  ITMO University\\
  \And
  Yuriy Dementiev\\
  ITMO University\\
  \And
  Artur Ignatiev\\
  ITMO University\\
  \And
  Danil Sagunov\\
  ITMO University\\
}
\date{}

\begin{document}

\maketitle

\begin{abstract}
Time-inconsistent behavior, such as procrastination or abandonment of long-term goals, arises when agents evaluate immediate outcomes disproportionately higher than future ones. This leads to globally suboptimal behavior, where plans are frequently revised or abandoned entirely. In the influential model of Kleinberg and Oren (2014) such behavior is modeled by a present-biased agent navigating a task graph toward a goal, making locally optimal decisions at each step based on discounted future costs. As a result, the agent may repeatedly deviate from initially intended plans.

Recent work by Belova et al. (2024) introduced a two-agent extension of this model, where a fully-aware principal attempts to guide the present-biased agent through a specific set of critical tasks without causing abandonment. This captures a rich class of principal--agent dynamics in behavioral settings.

In this paper, we provide a comprehensive algorithmic characterization of this problem. We analyze its computational complexity through the framework of parameterized algorithms, focusing on graph parameters that naturally emerge in this setting, such as treewidth, vertex cover, and feedback vertex set. Our main result is a fixed-parameter tractable algorithm when parameterized by the treewidth of the task graph and the number of distinct $(v,t)$-path costs.
Our algorithm encaptures several input settings, such as bounded edge costs and restricted task graph structure.
We demonstrate that our main result yields efficient algorithms for a number of such configurations.

We complement this with tight hardness results, that highlight the extreme difficulty of the problem even on simplest graphs with bounded number of nodes and constant parameter values, and motivate our choice of parameters.
We delineate tractable and intractable regions of the problem landscape, which include answers to open questions of Belova et al.\ (2024).
\end{abstract}

\section{Introduction}

Present bias---the tendency to overvalue immediate outcomes relative to future ones---is a well-studied phenomenon in behavioral economics~\cite{Laibson1997, Frederick2002}. It explains why individuals often abandon long-term beneficial plans in favor of short-term ease, leading to suboptimal results in domains such as health, finance, education, and productivity.

\cite{KleinbergO14} introduced a graph-theoretic model to formalize such behavior in sequential decision-making. In their model, a time-inconsistent agent with present bias navigates a directed acyclic graph (DAG) of tasks, choosing each step based on a discounted evaluation of future costs. The agent may abandon the plan before reaching the goal if their perceived cost exceeds the discounted reward. This model provides a powerful abstraction for reasoning about procrastination and dynamic decision-making in complex systems.

A natural extension of this framework, described in~\cite{ECAI24_BelovaDFGI24}, involves a principal (e.g., a teacher, system designer, or platform) who can intervene by modifying the graph to guide the agent toward desired outcomes. Such interventions may involve removing distracting tasks or adding helpful shortcuts. 
% Added arcs encode natural cases where the principal temporarily shares effort with the agent, reducing the agent’s perceived cost at the expense of the principal’s own resource. 
Our work focuses on this setting, where the principal seeks to ensure that a present-biased agent completes a sequence of important tasks and reaches the final goal.

% \textcolor{red}{It fits into a broader family of two-party computational models, such as GANs, prover–verifier systems, Merlin–Arthur protocols, and agent–mediator frameworks, where planning and control are separated and strategically coupled.}

\medskip\noindent\textbf{Problem Setting.}  
We study intervention strategies in the Kleinberg–Oren model, where an agent operates on a time-inconsistent planning model $M = (G, w, s, t, \beta, r)$:
\begin{itemize}
    \item $G = (V, E)$ is a directed acyclic graph with parallel arcs representing the task structure,
    \item $w: E \to \mathbb{N}_0$ assigns costs to arcs,
    \item $s$ and $t$ are the start and goal vertices,
    \item $\beta \in (0,1]$ is the present-bias factor,
    \item $r$ is the reward obtained upon reaching $t$.
\end{itemize}
When located at vertex $v$, the agent evaluates all $v$-$t$ paths and selects the one $P = e_1 e_2 \dots e_k$ that minimizes the \emph{perceived cost}
\[
\zeta_M(P) = w(e_1) + \beta \cdot \sum_{i=2}^k w(e_i).
\]
If $\zeta_M(P) > \beta \cdot r$, the agent abandons the task entirely. Otherwise, he commits to the first arc $e_1 = (v, u)$ and moves to vertex $u$, where the decision process is repeated recursively.
If several paths have a minimum perceived cost, the agent selects one of them deterministically, according to a given lexicographic order on the arcs.

Recent work by~\cite{ECAI24_BelovaDFGI24} initiated the study of \emph{principal-agent} problems, where a principal aims to ensure that the agent follows a path from $s$ to $t$ that includes a designated set of critical arcs $T \subseteq E(G)$, using limited graph modifications. Two types of interventions were considered:
\begin{itemize}
    \item \tpathdel: Can we delete at most $k$ arcs from $G$ so that the agent follows a path from $s$ to $t$ that includes all of $T$?
    \item \tpathadd: Given a set of auxiliary weighted arcs $A$, can we add at most $k$ arcs from $A$ (without creating cycles) so that the agent follows a path from $s$ to $t$ that includes all of $T$?
\end{itemize}

In this work, we unify these two problems into a single general formulation:

\defsimpleproblem{\tpathedit}{A time-inconsistent planning model $M = (G, w, s, t, \beta, r)$, a set of additional arcs $A$ such that $G+A$ is acyclic and $w$ maps $w$ maps $E(G)\cup A$ into $\mathbb{N}_0$, and a set of critical arcs $T \subseteq E(G)$.} {Compute the minimum number of arc edits (deletions from $E(G)$ or additions from $A$) needed so that the agent follows an $s$–$t$ path (path in $G$ from $s$ to $t$) that traverses all arcs in $T$.}
%\todo[inline]{Correctly define $A$}

Throughout this paper, we say that the agent follows a path $P$ if the agent, acting according to the Kleinberg–Oren model, traverses every arc of $P$ without abandoning the task at any intermediate vertex.

\begin{comment}
This problem formulation subsumes both \tpathdel and \tpathadd, and enables us to investigate structural aspects of the intervention problem through the lens of graph modification. 
\end{comment}
Notably, unlike the previous formulations that fix a budget $k$, the optimization version allows a more global view of tractability.

\medskip\noindent\textbf{Problem Motivation.}  
In emerging computational systems involving human or autonomous agents, effective task delegation and coordination increasingly rely on explicit modeling of agent behavior. Modern multi-agent architectures often follow hierarchical paradigms, where higher-level entities (principals, mediators) interact with time- or resource-constrained agents to ensure alignment with overarching objectives.

Agents in such systems may exhibit irrationalities, including time inconsistency or locally myopic behavior. Understanding how and when such agents can be steered toward desired behaviors---especially under minimal interventions---forms a key building block in the design of robust, interpretable, and strategically aligned agent systems. The \tpathedit problem formalizes one such scenario, capturing goal-aligned influence via constrained structural modifications.

\medskip\noindent\textbf{Model Example.}
To illustrate our framework, consider a scenario in which Bob, an AI engineer at a large tech company, is tasked with developing a production-ready LLM-based agent for enterprise document analysis. His ultimate goal is to deploy the agent after thorough benchmarking on realistic tasks (vertex $t$). There are multiple possible strategies to reach this goal, each composed of subtasks such as data preprocessing, model fine-tuning, prompt engineering, integration, and evaluation. These subtasks form a directed acyclic graph (see~\Cref{fig:example}), where each arc represents an action, and its weight corresponds to the estimated engineering effort or time cost.

However, Bob exhibits present-biased behavior: when evaluating a future plan, he gives full weight to the immediate next task, but discounts the remaining effort by a factor $\beta=1/2$. Starting at node $s$, he considers several possible execution paths:

\begin{itemize}
\item $P_1 = sbet,\ P_2 = scft$: thorough and reliable pipelines involving data cleaning, fine-tuning, and full evaluation: $\zeta_M(P_1) = \zeta_M(P_2) = 10 + \frac{10+10}{2} = 20$;
\item $P_3 = sadt$: a simplified path using default weights and partial testing: $\zeta_M(P_3) = 10 + \frac{8+8}{2} = 18$;
\item $P_4 = sbft$: a shortcut based on prompt-only adaptation with minimal validation: $\zeta_M(P_4) = 10 + \frac{2+10}{2} = 16$.
\end{itemize}
Bob chooses the path $P_4$ that minimizes the perceived cost, and it is easy to see that at each next vertex he will continue to choose this path, and will eventually follow $sbft$.

Now consider the role of Alice, Bob’s team lead and project owner. Her priority is to ensure that the critical evaluation step (arc $(b, e)$) is completed, as it directly affects the model’s reliability and regulatory compliance. Alice provides Bob with a promised reward of $r = 36$ upon successful deployment. However, due to Bob’s present bias, he perceives this as only $\beta \cdot r = 18$ at each decision point.

Suppose Alice simply disables the prompt-only shortcut by removing arc $(b, f)$, hoping to block the path $P_4$, the perceived cost at $s$ becomes $10 + \frac{8+8}{2} = 18 = \beta \cdot r$, which achieve only on path $P_3$. Bob would follow the arc $(s,a)$. In contrast, the path $P_1 = sbet$, which only includes the crucial arc $(b, e)$, has a perceived cost of $20$---too high for Bob to choose under present bias.

To resolve this, Alice applies a strategic edit: she introduces a new auxiliary arc $(a, b)$, allowing Bob to access path $P = sabet$. This new path includes the critical task $(b, e)$, and thanks to the adjusted structure, it has a lower perceived cost $1 + \frac{10+10}{2} =11$ at key decision point $a$. As a result, Bob follows $P = sabet$, completing all required stages and aligning his behavior with Alice’s objectives.

% This example illustrates the central idea of our work: with minimal structural interventions, a principal can align a present-biased agent’s behavior with critical goals---even when the agent would otherwise deviate. In general, deciding which interventions are sufficient---and how many edits are needed---is a challenging computational problem.
\begin{figure}[ht]
 \center{\includegraphics[scale=0.4]{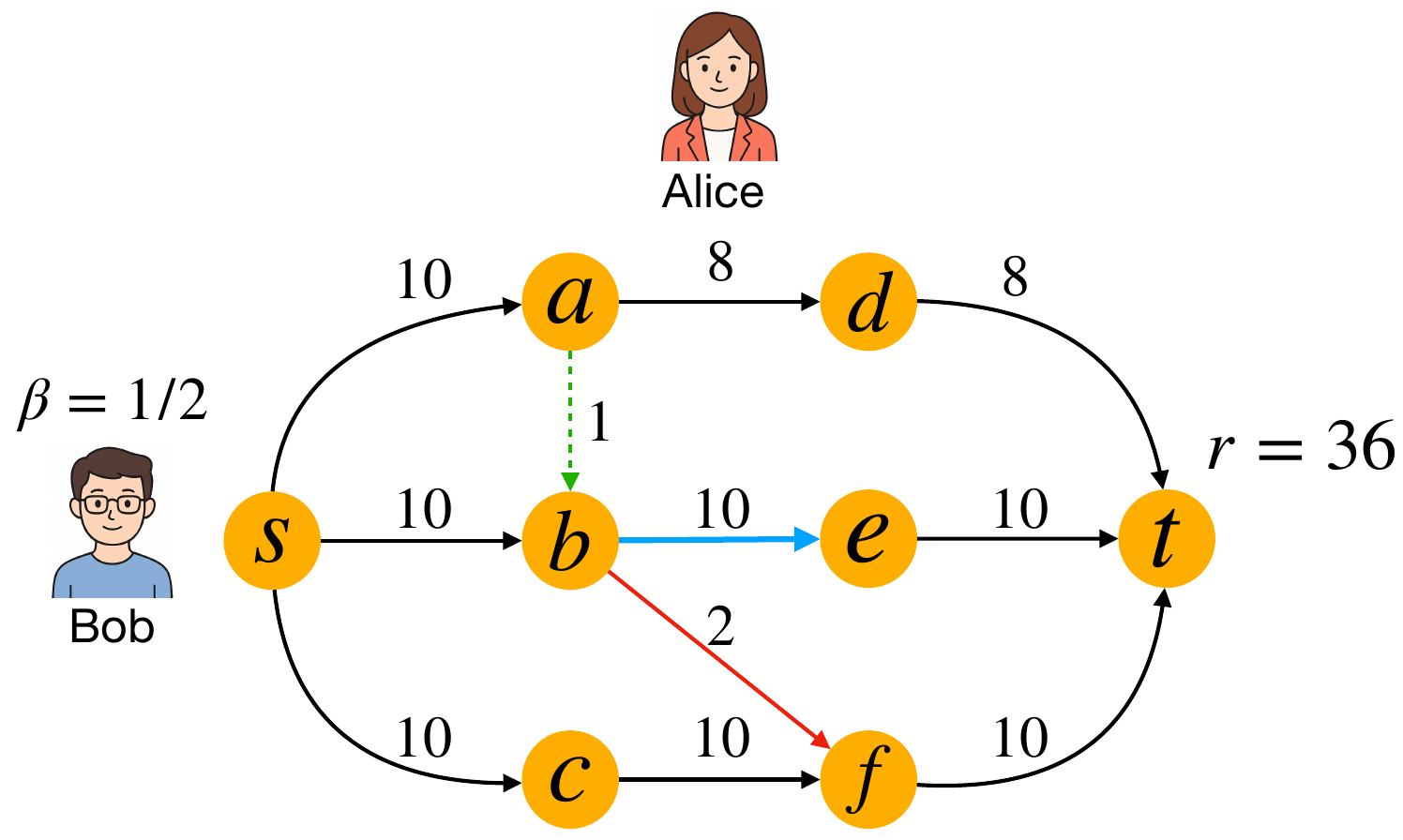}}
 \caption{Example: $T = \{(b, e)\},\ A = \{(a, b)\}$. 
 For initial graph $G$, the agent follows the path
 $sbft$.
 After deletion $(b, f)$---path $sadt$.
 And after adding $(a, b)$---$T$-path $sabet$.
 }\label{fig:example}
\end{figure}

\medskip\noindent\textbf{Why Graph Parameters Matter.}  
In real-world planning graphs modeling behavioral agents---such as workflows in education, research, digital assistants, or financial planning---the underlying structure is rarely arbitrary. Instead, it tends to reflect the limited complexity of human-designed or human-operated systems. This makes structural graph parameters highly relevant for understanding algorithmic tractability in practice~\cite{ICDT_ManiuSJ19}.

Below, we summarize the key graph parameters considered in our study.
All of these parameters are undirected graph parameters.
We consider them with respect to the underlying undirected graph of $G+A$.

The \emph{treewidth} ($\tw$) of a graph $G$ is the minimum width among all possible tree decompositions of $G$, where the width of a tree decomposition is defined as the size of its largest bag minus one. Intuitively, treewidth quantifies how “tree-like” a graph is. Many real-world planning graphs exhibit low treewidth due to hierarchical structure, and dynamic programming algorithms often become efficient on graphs of bounded treewidth~\cite{CyganFKLMPPS15}.

The \emph{feedback vertex set} number ($\fvs$) is the size of the smallest set of vertices whose removal renders the undirected graph skeleton acyclic. This parameter captures the extent to which cyclic dependencies exist and is often small in acyclic or near-acyclic planning systems.

The \emph{vertex cover} number ($\vc$) is the size of the smallest set $C \subseteq V(G)$ such that every edge $(u, v) \in E(G)$ is incident to at least one vertex in $C$. In task graphs with a few highly connected control points or verification steps, the vertex cover tends to be small.

The \emph{path length} ($p$) refers to the number of arcs on the longest directed path from source $s$ to target $t$ in the DAG. It reflects the maximum depth of planning or reasoning required by the agent.

The \emph{tree-depth} ($\td$) of a graph $G$ is the minimum height of a rooted tree $T$ such that $G$ is a subgraph of the closure of $T$; that is, for every edge $(u,v) \in E(G)$, one of $u$ or $v$ must be an ancestor of the other in $F$. Equivalently, tree-depth measures how ``deeply nested'' a graph is in terms of ancestor-descendant relations and reflects hierarchical complexity.

These parameters are not only natural from an algorithmic perspective, but also arise organically in structured task graphs.
Treewidth and feedback vertex set are small in many systems where tasks are layered or hierarchically decomposed—such as curriculum maps~\cite{CoursePrerequisite, CoursePrerequisite-2}, multi-stage workflows, or modular design pipelines. Even when cycles are present, they are typically localized and limited in scope.
Vertex cover tends to be small in graphs where dependencies are funneled through a small number of coordination or bottleneck nodes—common in multi-agent or supervisory planning setups~\cite{CyganFKLMPPS15}.
Path length is inherently bounded in most practical systems: no realistic agent executes hundreds of sequential steps before reaching a goal. In human-facing systems, this is constrained by attention spans, task fatigue, or real-time processing limits.

Taken together, these structural properties motivate a parameterized approach to the \tpathedit problem. By isolating parameters that remain small in realistic deployments, we identify meaningful tractable regimes where efficient intervention strategies for present-biased agents can be computed. This connection between structure and tractability bridges theoretical models with real-world applicability.

\medskip\noindent\textbf{Parameterized complexity.}  
To rigorously analyze the computational complexity of the \tpathedit problem, we rely on the framework of parameterized complexity. This framework allows us to distinguish between sources of hardness by isolating specific aspects of the input—such as structural graph parameters—as formal \emph{parameters}.

Formally, a parameterized problem is a language $Q \subseteq \Sigma^* \times \mathbb{N}$, where an instance is given as a pair $(I, k)$ consisting of the main input $I$ and a parameter $k$. A problem is said to be \emph{fixed-parameter tractable} (\classFPT) if it can be solved in time $f(k) \cdot |I|^{\Oh(1)}$ for some computable function $f$ depending only on $k$. This contrasts with the broader class \classXP, where the running time is allowed to be $|I|^{g(k)}$ for some computable function $g$.

To classify intractable problems, the $\operatorClassW$-hierarchy provides a series of hardness classes:
$
\classFPT = \classW{0} \subseteq \classW{1} \subseteq \classW{2} \subseteq \cdots \subseteq \classW{P}.
$
It is widely believed that $\classFPT \neq \classW{1}$, so proving $\classW{1}$-hardness is strong evidence that no fixed-parameter tractable algorithm exists. Such hardness proofs are typically obtained via parameterized reductions from known $\classW{1}$-complete problems.

In addition, we use the notion of \classParaNP-hardness, which refers to parameterized problems that remain \classNP-hard even when the parameter is fixed to a constant. This notion is useful for ruling out tractability in low-parameter regimes.

For a general reference on the theory and techniques of parameterized complexity, we refer the reader to~\cite{CyganFKLMPPS15}.

\medskip\noindent\textbf{Our Contribution.} 
We present algorithmic results for \tpathedit, highlighting structural and numerical conditions under which efficient interventions are possible. We comprehensively describe a complexity landscape of these problems. Our main result is an algorithm parameterized by the treewidth ($\tw$) of the graph $G+A$ and the diversity of path costs:

\begin{restatable}{theorem}{mainResult}\label{thm:main-result}
\tpathedit admits an algorithm running in $|L|^{\mathcal{O}(\tw)} \cdot m^{\Oh(1)}$ time, where $\tw$ equals the treewidth of $G+A$ and
\[
L = \left| \bigcup_{v \in V(G)} \{ w(P) \mid P \text{ is a } (v,t)\text{-path in } G+A  \} \right|.
\]
\end{restatable}

This shows that \tpathedit is fixed-parameter tractable when both the graph structure (via treewidth $\tw$) and weight diversity are bounded.
We also show:
\begin{itemize}
    \item \classNP-hardness on graphs of treewidth $2$,
    \item \classW{1}-hardness w.r.t.\ the number of vertices $n$ and the combined parameter $(p, \vc)$,
    \item \classXP-algorithms with respect to feedback vertex set, vertex cover, and tree-depth,
    \item \classFPT-algorithms for cost-limited instances with bounded cost diversity $|\{w(e)\}|$.
\end{itemize}

\medskip\noindent\textbf{Related Work.}  
The study of present-biased behavior originates from classical economic theory. Samuelson’s discounted utility model~\cite{Samuelson1937} laid the foundation for formal models of intertemporal choice. This was later extended into the hyperbolic and quasi-hyperbolic discounting frameworks~\cite{Laibson1994,McClure2004}, which more accurately capture behavioral inconsistencies over time. The model we adopt follows the framework introduced by Kleinberg and Oren~\cite{KleinbergO14,KleinbergO18}, where a present-biased agent navigates a task graph by selecting paths according to discounted cost evaluations. Their model can be viewed as a discrete analogue of quasi-hyperbolic discounting and incorporates Akerlof’s notion of salience~\cite{Akerlof91}.

While the Kleinberg–Oren model has strong empirical motivation, it does not capture all known psychological aspects of time-inconsistent behavior~\cite{Frederick2002}. Nevertheless, it has spurred extensive algorithmic research, including complexity and approximation analyses~\cite{GravinILP16,KleinbergOR16,KleinbergOR17,halpern2023chunking,meyer2022present,AAAI22_DFI, FominS20}.

Several recent works have extended the Kleinberg–Oren model in complementary directions. \cite{aaai/AkagiMK24, aaai/AkagiKK25} propose a continuous-time variant with closed-form solutions for optimal reward placement and abandonment conditions. \cite{ECAI24_BelovaDFGI24} introduce the \tpathedit framework and analyze structural interventions in two-agent settings.

The principal–agent formulation we study can be seen as a natural and timely generalization of the classical motivating subgraph and $P$-motivating subgraph problems, where the goal is to identify a subgraph that leads a present-biased agent to the goal (through a specific path $P$) without abandonment.~\cite{tang2017computational, AlbersK17, Albers2018, oren2019principal, FominS20}.

Finally, our work is also situated in the broader context of graph modification problems, a well-established domain in algorithmic graph theory. These problems aim to alter a graph's structure—via edge or vertex deletions/additions—to satisfy desired properties. We refer the reader to surveys such as~\cite{burzyn2006np,natanzon2001complexity,CrespelleDFG23} for comprehensive overviews.

\begin{comment}
Our contribution complements and extends this body of work by studying behavioral planning under structural constraints, enriching both algorithmic behavioral economics and graph modification theory.
\end{comment}

\section{Polynomial Algorithm on Bounded Vertex Cover Graphs}

In this section, we prove that \tpathedit can be solved in polynomial time when vertex cover of $G+A$ has constant size.

The first observation about graphs with bounded vertex cover is that such graphs have bounded path lengths.

%\todo[inline]{indicate star}
\begin{restatable}{observation}{obser}\label{obs:p-by-vc}
% \footnote{By ($\star$) we indicate results with proofs deferred to the appendix due to the space constraints.}
In a DAG $G$, every path consists of at most $2\cdot\vc(G)-1$ arcs.
\end{restatable}
\begin{proof}
    Let $S$ be the minimum vertex cover of $G$, that is, any arc in $G$ either starts in $S$ or ends in $S$ or both.
    We have $|S|=\vc(G)$.

    Targeting towards a contradiction, assume that $G$ has a path $P$ with at least $2\vc(G)$ arcs.
    Since $G$ is acyclic, this path is simple and contains at least $2\vc(G)+1=2|S|+1$ unique vertices.
    By pigeon-hole principle, there are two consecutive vertices on $P$ that do not belong to $S$.
    That is, there is an arc in $G$ that starts outside $S$ and ends outside $S$.
    The obtained contradiction concludes the proof.
\end{proof}

Before moving on to the main result of this section, we show an auxiliary algorithmic result.
We will use it as a subroutine in our main algorithm.

\begin{restatable}{proposition}{prop}\label{prop:bipartite-vertex-cover}
    There is a polynomial-time algorithm that,
    given two disjoint sets $A$, $B$ and a sequence of $m$ pairs $(a_1,b_1),(a_2,b_2),\ldots, (a_m, b_m)\in A\times B$, and a set $R\subset (A\cup B)$ finds minimum possible $X\subset (A\cup B)$ such that
    \begin{itemize}
        \item for each $i\in [m]$, $\{a_i,b_i\}\cap X\neq \emptyset$ and
        \item $X\cap R=\emptyset$,
    \end{itemize}
    or reports that $X$ does not exist.
\end{restatable}

\begin{proof}
    We present a polynomial-time algorithm.
    First, the algorithm determines whether $X$ exists.
    If there is a pair $(a_i, b_i)$ with both $a_i, b_i\in R$, then $X$ cannot satisfy both constraints simultaneously.
    In this case, the algorithm reports that $X$ does not exist and stops. 

    The next part of the algorithm consists of two stages.
    In its first stage, the algorithm constructs a set $X'\subset A\cup B$ of elements that should belong to any valid solution $X$.

    The set $X'$ is constructed iteratively until $R$ is empty.
    When $R$ is not empty, there is a element $a\in R$.
    Without loss of generality, $a\in A$.
    Then for each pair $(a_i, b_i)$ with $a_i=a$, the algorithm takes $b_i$ into $X'$, since the only possible way to hit $\{a_i, b_i\}$ is to take $b_i$.
    Then the algorithm removes $a$ from $A$ and $R$ and removes all pairs that contain $a$.
    If $R$ is not empty, the algorithm repeats the iteration with the next element of $R$.
    Otherwise, the algorithm proceeds to the second stage.
    The case $a\in B$ is symmetrical.

    In its second stage, the algorithm first removes all pairs $(a_i, b_i)$ with $|\{a_i, b_i\}\cap X'|>1$, since these pairs all already hit by $X'$.
    Then, the algorithm removes all vertices of $X'$ from $A$ and $B$.

    After such preprocessing, the algorithm has $R=\emptyset$, $A\cap X'=B\cap X'=\emptyset$, and a remaining list of pairs with $(a_i,b_i)\in A\times B$.
    The remaining task of the algorithm is to find $X''\subset (A\cup B)$ of minimum size that intersects each of the remaining pairs, without any additional restrictions.
    
    This problem is equivalent to \textsc{Bipartite Vertex Cover} problem on a bipartite graph $G$ with vertex set $A\cup B$, and edge set equal to the set of remaining pairs.
    There is a one-to-one correspondence between a vertex cover in $G$ and a subset of $A\cup B$ hitting each of the remaining pairs.
    Minimum vertex cover in a bipartite graph can be found in polynomial time using maximum bipartite matching algorithm together with a constructive proof of K\H{o}nig's theorem.
    For more information on this topic, we refer the reader to classical algorithm books \cite{KleinbergTardos} and \cite{Cormen-book}.

    The algorithm constructs $G$ in polynomial time, and finds the minimum vertex cover $X''$, which equals the minimum cardinality set hitting all of the remaining pairs.
    Finally, the algorithm outputs $X=X'\cup X''$ and stops.

    Correctness of the algorithm follows from the discussion.
    The proof is complete.
\end{proof}

We note that \Cref{prop:bipartite-vertex-cover} is folklore and is essentially a problem of finding a vertex cover of minimum cardinality in a bipartite graph.
For completeness, we provide the proof with necessary references 
\ifshort
in the appendix.
\else
below.
\fi

We now show the XP-algorithm for \tpathedit parameterized by $\vc$.
In the following proof, our algorithm makes comparisons between perceived costs.

\begin{theorem}\label{thm:vertex-cover}
  \tpathedit admits an algorithm with running time $m^{\Oh(\vc^2)}$, where $\vc$ is the vertex cover number of $G+A$.
\end{theorem}
\begin{proof}

    We present an algorithm that finds an optimal solution to the given instance $((G,w,s,t,\beta,r), A, T)$ of \tpathedit.
    If no solution exists, the algorithm correctly reports it.

    We assume without loss of generality that the algorithm is additionally given a vertex cover $C$ or $G+A$ with $|C|\le 2\cdot \vc(G+A)$, as such vertex cover can be found in polynomial time (see, e.g.,  \cite{PapadimitriouS82} for $2$-approximation of \textsc{Vertex Cover}).

    Let $G^*$ be the graph obtained via optimal sequence of arc edits.
    Speaking informally, our algorithm does not know $G^*$ in advance, so its goal is to guess only the \emph{relevant} structure of $G^*$: the agent's $s$-$t$ path and \emph{important} shortest paths in $G^*$.
    Since all path length are bounded in $G+A$, (therefore in $G^*$), the search space for these guesses is bounded with $m^{\Oh(\vc^2)}$.
   The high-level sketch of our algorithm is the following.
    \begin{enumerate}
        \item Iterate over all $s$-$t$ paths $P$ in $G+A$ (the guess for the agent's path);
        \item Construct a set $S$ consisting of all vertices of $C$ and all vertices of the path $P$. That is, each arc of $P$ has both endpoints in $S$ and each arc of $G$ starts in $S$ or ends in $S$ or both;
        \item For each $v\in S$, iterate \emph{one} path $R_v$ from $v$ to $t$ in $G+A$ (the guess for the shortest path);
        \item Find a smallest possible set of arc deletions that 
        \begin{itemize}
            \item forces the agent to follow $P$, and \item for each $v\in S$, breaks every $v$-$t$ path shorter than $R_v$.
        \end{itemize}
    \end{enumerate}

    We move on to the formal part of the proof.
    We describe the steps of the algorithm formally and supply it with correctness claims when needed.

    \medskip\noindent\textbf{Iterating the agent's $s$-$t$ path and critical vertices.}
    In its outer loop, the algorithm iterates an arc set $P\subset (E(G)\cup A)$ such that $P\supset T$ and $P$ forms a simple $s$-$t$ path in $G+A$.
    In particular, $|P|\le 2|C|-1$.

    Slightly abusing the notation, we refer to $P$ both as an arc set and a path in $G+A$.
    Let $d$ be the length of $P$, and let $u_0, u_1, \ldots, u_d$ be the vertices on $P$ in the natural order, so $u_0=s$ and $u_d=t$.
    For each $i\in [d]$, the algorithm denotes by $e_i$ the edge between $u_{i-1}$ and $u_i$ in $P$.

    The algorithm then constructs a set $S:=\{u_0, u_1,\ldots, u_d\}\cup C$ of \emph{critical} vertices.

    \medskip\noindent\textbf{Iterating the critical shortest paths edges.}
   The algorithm then iterates $R\subset (E(G)\cup A)$ such that $|R|\le (|S|+1)\cdot (2|C|-1)$, and $R\supset P$ and
    \begin{enumerate}[label=(R\arabic*), ref=(R\arabic*), leftmargin=*]
        \item\label{prop:does-not-abandon} for each $i\in [d]$, $w(e_i)+\beta\cdot \dist_R(u_i, t)\le \beta \cdot r$.
    \end{enumerate}
     
     Here and further in this proof, we slightly abuse the notation and denote by $\dist_R(x,y)$ the distance between $x$ and $y$ in the graph $(V(G), R)$.

     An intuition behind $R$ is that it is (when guessed correctly) a union of all  shortest paths that are required for the agent to follow $P$.
     \ref{prop:does-not-abandon} specifically ensures that the agent does not abandon $P$.

    \medskip\noindent\textbf{Highlighting local obstructions.}
    Having $P$ and $R$ both fixed, the algorithm now aims to find the minimum size subset $X\subset E(G)$ (the edges one has to delete) such that, for $G'=(V(G), R\cup (E(G)\setminus X))$:
        (a) we do not delete any edge in $R$;
        (b) distances in $G'$ agree with distances in $R$;
        (c) the agent follows every edge $e_i\in P$ in $G'$.
    Expressed formally, this becomes
    \begin{enumerate}[label=(X\arabic*), ref=(X\arabic*), leftmargin=*]
        \item \label{prop:no-delete-r}  $X\cap R=\emptyset$; 
        \item\label{prop:distances-implemented}  For each $v\in S$, $\dist_{G'}(v,t)=\dist_{R}(v,t)$;
        \item\label{prop:no-turn-aways} For each $i\in [d]$, for each vertex $v\in V(G)$ and for each edge $e'\in E(G')$ that goes from $u_{i-1}$ to $v$ and $e'\neq e_i$, holds $$w(e_i)+\beta\cdot \dist_{G'}(u_{i},t)\prec w(e')+\beta\cdot\dist_{G'}(v,t).$$
    \end{enumerate}

    Here and forth, we use $\prec$ to indicate the comparison between perceived costs according to the agent's behavior.
    That is, if the parts to the left and to the right of $\prec$ are equal, the result of comparison is equivalent to the comparison of $e_i$ and $e'$ according to the ordering of edges.

    We claim that these constraints are essentially equivalent to the agent following $P$ in $G'$. 

    \begin{claim}\label{claim:equiv}
    Let $X\subset E(G)$ satisfy \ref{prop:distances-implemented}.
    Then
        $R$ satisfies \ref{prop:does-not-abandon} and $X$ satisfies  \ref{prop:no-turn-aways} if and only  the agent follows $P$ in $G'$.
    \end{claim}
    \begin{claimproof}
        Since \ref{prop:distances-implemented} holds, $\dist_{R}(v,t)$ and $\dist_{G'}(v,t)$ are equivalent if $v\in S$.
        The agent follows $P$ in $G'$ if and only if:
        \begin{itemize}
            \item The agent does not abandon the task in neither of $u_0, u_1, \ldots, u_{d-1}$, so for each $i\in [d]$ we have \linebreak
            $w(e_i)+\beta\cdot \dist_{G'}(u_i,t)\le  \beta\cdot r.$
            This is equivalent to \ref{prop:does-not-abandon}.
            \item For each $i\in [d]$, the agent always prefers the edge $e_i$ when goes from $u_{i-1}$, meaning taking any other edge $e'$ outgoing from $u_{i-1}$ gives greater cost estimations (or estimations for $e'$ and $e_i$ are equal but $e'$ goes after $e_i$ in the ordering of edges).
            This is equivalent to \ref{prop:no-turn-aways}.
        \end{itemize}
        The equivalences prove the claim.
    \end{claimproof}

    We now explain how the algorithm finds suitable $X$ algorithmically.

    \medskip\noindent\textbf{Obstructions between critical vertices.}
    The algorithm first constructs $X_1$ by finding the arcs with critical endpoints that have to be deleted necessarily to fulfill the properties.
    $X_1$ consists of all arcs $e'$ of $G$ such that $e'$ starts in some $x\in S$ and ends in some $y\in S$, and either
        \begin{enumerate}[label=(\roman*), leftmargin=1.8em]
            \item $\dist_{R}(x,t)>w(e')+\dist_R(y,t)$, or
            \item $x=u_{i-1}$ for some $i\in [d]$, $e'\neq e_i$ and
            $w(e_i)+\beta\cdot \dist_R(u_i, t)\succ w(e')+\beta\cdot \dist_R(y, t).$
        \end{enumerate}

    We claim that all arcs in $X_1$ are necessary to delete.

    \begin{claim}\label{claim:x1-subset-of-x}
        If $X$ satisfies \ref{prop:distances-implemented}, \ref{prop:no-turn-aways} then $X_1\subset X$.
    \end{claim}
    \begin{claimproof}
        Let $X\subset E(G)$ satisfy \ref{prop:distances-implemented}, \ref{prop:no-turn-aways}.
        Targeting towards a contradiction, suppose $e'\in X_1\setminus X$.
        Denote by $x,y\in S$ the start and end points of $e'$, i.e.\ $e'$ goes from $x\in S$ to $y\in S$.
        
        If (i) is true, then $\dist_R(x,t)>w(e')+\dist_R(y,t)$, then by \ref{prop:distances-implemented} we have $\dist_{G'}(x,t)>w(e')+\dist_{G'}(y,t)$, which contradicts the definition of shortest distances, as $e'$ goes from $x$ to $y$.

        If (ii) is true, then by \ref{prop:distances-implemented} we have
        $$w(e_i)+\beta\cdot \dist_{G'}(u_i,t)\succ w(e')+\beta\cdot \dist_{G'}(y,t),$$
        and this clearly violates \ref{prop:no-turn-aways}, since both $e_i$ and $e'$ start in $u_{i-1}$.
        The obtained contradiction finishes the proof of the claim.
    \end{claimproof}

    We will show a little later that in any optimal solution, there is no arc between vertices in $S$ that should be deleted but does not belong to $X_1$.

    \medskip
    \noindent\textbf{Obstructions incident to non-critical vertices.}
    Then, for each non-critical vertex $v\in V(G)\setminus S$, the algorithm constructs an arc pair set $\mathcal{X}_v$.
    This is done in the following way.
    Consider an arc $e'$ that ends in $v$ and an arc $e''$ that starts in $v$.
    Let $x$ and $y$ be the start point of $e'$ and the end point of $e''$ respectively.
    Note that $x,y\in S$, since we cannot have arcs between two non-critical vertices.
    The pair $\{e',e''\}$ is added to $\mathcal{X}_v$ if
    \begin{enumerate}[label=(\roman*), leftmargin=1.8em]
        \item $\dist_R(x,t)>w(e')+w(e'')+\dist_R(y,t)$, or
        \item $x=u_{i-1}$ for some $i\in [d]$ and
    $w(e_i)+ \beta \dist_R(u_i, t)\succ w(e')+\beta (w(e'')+ \dist_R(y, t)).$
    \end{enumerate}

    We show that its necessary for $X$ to hit every pair in $\mathcal{X}_v$.
    \begin{claim}\label{claim:x-hits-pairs}
        Let $X\subset E(G)$ satisfy \ref{prop:distances-implemented}, \ref{prop:no-turn-aways}.
        Then, for each $v\in V(G)\setminus S$ and for each $Y\in \mathcal{X}_v$, we have $X\cap Y\neq \emptyset.$
    \end{claim}
    \begin{claimproof}
        The proof resembles the proof of \Cref{claim:x1-subset-of-x} with a pair of edges $e',e''$ instead of a single edge $e'$.
        Targeting towards a contradiction, assume there exists $v\in V(G)\setminus S$ and $\{e',e''\}\in \mathcal{X}_v$ such that $e',e''\notin X$.

        If $\{e',e''\}$ satisfies (i), then by \ref{prop:distances-implemented} $$\dist_{G'}(x,t)\succ w(e')+w(e'')+\dist_{G'}(y,t),$$ which contradicts the definition shortest distances identically to the proof of \Cref{claim:x1-subset-of-x}.

        The remaining case is when $\{e',e''\}$ satisfies (ii).
        By \ref{prop:distances-implemented} and $u_i, y\in S$ we have $\dist_R(y,t)=\dist_{G'}(y,t)$ and $\dist_R(u_i,t)=\dist_{G'}(u_i,t).$
        Since $e''$ goes from $v$ to $y$, we have $\dist_{G'}(v,t)\le w(e'')+ \dist_{G'}(y,t)$ by definition of $\dist$.
        Combining all of these with (ii) yields
        $$w(e_i)+\beta\cdot \dist_{G'}(u_i,t)\succ w(e')+\beta\cdot \dist_{G'}(v,t),$$
        contradiciting \ref{prop:no-turn-aways}.
        The proof of the claim is complete.
    \end{claimproof}

    \noindent\textbf{Inner step of the algorithm.}
    Having $P$ and $R$ fixed by the two outer loops, and $X_1$ and $\mathcal{X}_v$ constructed for each $v\notin S$, the algorithm finds the set $X$ of minimum cardinality that satisfies \ref{prop:no-delete-r} and the right parts \Cref{claim:x1-subset-of-x} and \Cref{claim:x-hits-pairs}, that is, $X$ does not intersect $R$,
    $X$ contains $X_1$, and $X$ intersects every pair in $\bigcup_{v\notin S} \mathcal{X}_v$.
    
    The algorithm might also determine that $X$ does not exist (this would mean basically that the choice of $P$ and $R$ is wrong, or no solution exists at all).
    This happens, for instance, when $X_1\cap R\neq \emptyset$.

    Since $\mathcal{X}_v$ and $\mathcal{X}_{v'}$ contains pairwise distinct arcs, the minimum possible set intersecting every pair of $\mathcal{X}_v$ can be done independently for each vertex $v \notin S$.
    For each $v\notin S$, the algorithm uses polynomial-time subroutine \Cref{prop:bipartite-vertex-cover} to find minimum-size set $X_v$ that hits every pair in $\mathcal{X}_v$ and does not intersect $R$.
    If the subroutine of \Cref{prop:bipartite-vertex-cover} reports that $X_v$ does not exist, the algorithm correctly determines that $X$ does not exist and moves to the next iteration of the outer loops.

    If $X_v$ was successfully constructed for each $v\notin S$, the algorithm constructs $X$ via $X:=X_1\cup\bigcup_{v\in V(G)\setminus S} X_v$.
    The algorithm obtains a solution of size $|R\cap A|+|X|$: edges in $(R\cap A)$ should be added to $G$, edges in $X$ should be removed from $G$.
    The algorithm updates its best solution with this one and moves to the next choice of $P$ and $R$.

    The description of the algorithm is finished.
    We move on to proving its correctness.

 \medskip   \noindent\textbf{The algorithm never outputs incorrect solutions.}
    We first show that $(R\cap A)\cup X$ is always a correct solution.
    Note that if we add  $R\cap A$ to $G$ and remove $X$ from $G$, we obtain exactly $G'=(V(G), R\cup (E(G)\setminus X)$.
    Therefore, it is enough to prove that the agent always follows $P$ in $G'$.
    
    We show that $X$ constructed in the inner step of the algorithm satisfies \ref{prop:distances-implemented} and \ref{prop:no-turn-aways}.

    \begin{restatable}{claim}{clai}\label{claim:x-satisfies}
        If $X$ satisfies \ref{prop:no-delete-r}, $X$ contains $X_1$ and $X$ intersects every pair of  $\bigcup_{v\notin S}\mathcal{X}_v$, then $X$ satisfies \ref{prop:distances-implemented} and \ref{prop:no-turn-aways}.
    \end{restatable}
        \begin{claimproof}
        The proof is by contradiction.
        We first prove \ref{prop:distances-implemented} and then prove \ref{prop:no-turn-aways} using that $X$ satisfies \ref{prop:distances-implemented}.
        
        \textbf{$X$ satisfies \ref{prop:distances-implemented}}. Assume that $X$ does not satisfy \ref{prop:distances-implemented} (the distances in $G'$ does not agree with $R$).
        This means that there exists $x\in S$ with $\dist_{G'}(x,t)\neq \dist_R(x,t)$.
        Among all such possible choices of $x$, choose $x$ topologically as close to $t$ as possible.
        Since $R\subset E(G')$, every path with arcs from $R$ is present in $G'$.
        Consequently, we have that $\dist_{G'}(x,t)<\dist_R(x,t).$
        Note that the path from $x$ to $t$ has at least one arc since $\dist_{G'}(t,t)=\dist_R(t,t)=0$.
        
        Let $v$ be the first vertex after $x$ on the shortest $x$-$t$ path in $G'$ , and let $e'$ be the first arc on this path  ($e'$ goes from $x$ to $v$).
        If $v$ is critical, that is, $v\in S$, then $\dist_R(v,t)=\dist_{G'}(v,t)$ by definition of $x$.
        Then we have
        \[
             \dist_R(x,t)>\dist_{G'}(x,t) =w(e')+\dist_{G'}(v,t) =w(e')+\dist_{R}(v,t).
        \]
        That is, $e'$ satisfies (i) in construction of $X_1$, so $e'\in X_1\subset X$.
        But $e'\in E(G')=R\cup (E(G)\setminus X)$, so $e'\in R$.
        This contradicts \ref{prop:no-delete-r}.

        When $v\notin S$, we have that there is a second arc $e''$ and a third vertex $y$ on the shortest $x$-$t$ path in $G'$.
        The vertex $y$ belongs to $S$, because there are no arcs between vertices outside $S$.
        Hence, $\dist_R(y,t)=\dist_{G'}(y,t)$.
        Proceeding similarly to as before in this proof, we have
        $$\dist_R(x,t)>w(e')+w(e'')+\dist_R(y,t),$$
        meaning that $\{e',e''\}\in \mathcal{X}_v$, so at least one of $e',e''$ belongs to $X$.
        But both these edges belong to $E(G')=R\cup (E(G)\setminus X)$, meaning that $R\supset (X\cap \{e',e''\})$, contradicting (X1) as before.

        \textbf{$X$ satisfies \ref{prop:no-turn-aways}}.
        Assume that $X$ does not satisfy \ref{prop:no-turn-aways} (the agent turns away from $P$ in $G'$).
        That is, there exists $i\in [d]$ and an arc $e'\in E(G')$ going from $u_{i-1}$ to some $v\in V(G)$ with $e'\neq e_i$ and
        \begin{equation}\label{eq:temp1}
        w(e_i)+\beta\cdot \dist_{R}(u_i,t)\succ w(e')+\beta\cdot \dist_{G'}(v,t),
        \end{equation}
        where $\dist_{R}(u_i,t)=\dist_{G'}(u_i,t)$ follows from $u_i\in S$ and $X$ satisfying \ref{prop:distances-implemented}.
        
        If $v\in S$, then we can as well rewrite $\dist_{G'}(y,t)$ with $\dist_R(y,t)$ in \eqref{eq:temp1}.
        Hence, $e'$ satisfies (ii) in the construction of $X_1$.
        Then $e'\in X_1\subset X$.
        This contradicts \ref{prop:no-delete-r} exactly as before in the proof of this claim.

        The remaining case is $v\notin S$.
        The shortest $v$-$t$ path in $G'$ has at least one arc vertex, let $e''$ be the first arc on this path and let $y$ be the second vertex on this path ($e''$ goes from $v$ to $y$). 
        We have that $y\in S$.
        Then 
        $\dist_{G'}(v,t)=w(e'')+\dist_{G'}(y,t)=w(e'')+\dist_R(y,t).$
        Combining this with \eqref{eq:temp1} gives that $e',e''$ satisfy (ii) in the construction of $\mathcal{X}_v$.
        Consequently, $\{e',e''\}\cap X\neq \emptyset$, which contradicts \ref{prop:no-delete-r}.

        The proof of the claim is complete.
    \end{claimproof}

    Combining \Cref{claim:x-satisfies} and \Cref{claim:equiv}, we obtain that the agent follows $P$ in $G'$, and $P$ traverses all arcs $T$ by construction.
    Therefore, the algorithm outputs only correct solutions.
    
\medskip   \noindent\textbf{The algorithm finds a minimum-size solution.}
To see that the algorithm outputs a solution of minimum size, we show a valid choice of $P$ and $R$ that leads to the optimal answer.

To see this choice, let $G^*$ be the graph obtained from $G$ with the optimal edit sequence.
Let $P^*$ be the agent's $s$-$t$ path in $G^*$.
We have that $T\subset P^*$, and $|P^*|\le 2\vc(G+A)-1\le 2|C|-1$.
Let $S$ be the union of all vertices of $P^*$ and $C$.
For each $v\in S$, let $R^*_v$ be the edge set of the shortest path between $v$ and $t$ in $G^*$.
Construct $R^*$ as a union of $P^*$ and $\bigcup_{v\in S}R^*_v$, clearly $$|R^*|\le (2|C|-1)+|S|\cdot (2|C|-1)\le (|S|+1)\cdot (2|C|-1),$$
and $R^*$ satisfies \ref{prop:does-not-abandon} since the agent does not abandon $P^*$ in $G^*$.

Therefore, $P^*$ and $R^*$ are valid choices of $P$ and $R$ that the algorithm will consider in one of its iterations.
Note that $X^*$ (a sequence of arc deletions for $G^*$) satisfies \ref{prop:no-delete-r},\ref{prop:distances-implemented}, \ref{prop:no-turn-aways} for this choice of $P$ and $R$.
Hence, the algorithm will find edge deletion set $X$ with $|X|=|X^*|$ and update the answer with an arc edit set of size $|R^*\cap A|+|X^*|$, which is minimum possible by definition of $G^*$.

\medskip\noindent\textbf{Running time analysis.}
The inner step of the algorithm is polynomial in $m$, so we have to analyze the number of iteration given by its two loops.
In its first loop, the algorithm iterates arc sets of size at most $\mathcal{O}(\vc)$, and in its second loop it iterates arc sets of size $\mathcal{O}(\vc^2)$.
Multiplied, this gives a total of $m^{\Oh(\vc^2)}$ iterations.
%The proof is complete.
\end{proof}

\section{Parameter Landscape}

The proof of our main algorithmic result, \Cref{thm:main-result}, shares its basic idea with the proof of \Cref{thm:vertex-cover}: guess each distance from $v$ to $t$, ensure that necessary arcs are added, and delete necessary arcs so that the resulting distances agree with the guess and the resulting agent's path traverses each arc in $T$.
To achieve the running time of \Cref{thm:main-result}, we avoid guessing the agent's path $P$, as we did in \Cref{thm:vertex-cover}, because this approach requires comparing arbitrary arcs of $G+A$ lexicographically.
In the proof of \Cref{thm:main-result}, we highly rely on that it's enough to compare arcs of $G+A$ only to arcs from $T$.

Other than that, \Cref{thm:main-result} is a technical dynamic programming algorithm over a tree decomposition, common to the field of parameterized complexity (see, e.g, \cite{CyganFKLMPPS15}).
The complete proof of \Cref{thm:main-result} can be found in the appendix~\ref{appendix}.

We now demonstrate that \Cref{thm:main-result} captures several structural settings for \tpathedit, including the bounded vertex cover scenario and constant arc costs.

\begin{restatable}{lemma}{lemm}\label{lem:corollaries}
    Let $(G,w)$ be a DAG with non-negative integer arc costs $w: E(G)\to W$.
    Let $t$ be a vertex in $G$.
    Let $p$ be a maximum number of arcs in a path in $G$.
    Then $$|L|\le \min\{(p+1)^{|W|},1+p\cdot \max W\} \text{ and}$$
    $$|L|\le \min\{m^{p+1}, m^{2\cdot\vc},m^{2^{{\td+1}}}, m^{2\cdot \fvs+1}\}.$$
\end{restatable}
 \begin{proof}
    To see the first part of the lemma, note that any path in $(G,w)$ consists of at most $p$ arcs.
    Hence, a cost of a path in $(G,w)$ is a sum of $p$ integers taken from $W$.
    The cost of a path $P$ can be determined by the number of arcs of cost $w$ it contains, for each $w\in W$ (at most $(p+1)^{|W|}$ combinations).
    This sum is an integer that can be as low as $0$ (empty path) or as high as $p\cdot \max W$ ($p$ arcs with maximum possible cost). $|L|\le \min\{(p+1)\cdot W, 1+p\cdot \max W\}$ follows.

    To see the second part, note that any path $P$ in $G$ is a set of at most $p$ arcs, and the number of such sets is at most $\sum_{i=0}^{p}\binom{m}{i}\le m^{p+1}$.
    For vertex cover, we have that $p\le 2\cdot \vc-1$ (\Cref{obs:p-by-vc}).
    For tree-depth, we have that $p\le 2^{\td+1}-2$ (see, e.g., \cite{Demaine2019}).

    A graph with bounded feedback vertex set size can have paths of unbounded lengths.
    But any $(v,t)$-path in $G$ is uniquely defined by its starting arc ($m$ options), and $0$ or $2$ arcs per a vertex of the feedback vertex set ($(1+\binom{m}{2})^{\fvs}$ options), forming a total of at most $m^{2\fvs+1}$ paths.
\end{proof}

\Cref{lem:corollaries} has interesting consequences.
Because each of the three parameters $\vc$, $\fvs$, $\td$ is at least $\tw$, the running time of the algorithm of \Cref{thm:main-result} is upper-bounded by $m^{\Oh(\vc^2)}$, $m^{\Oh(\fvs^2)}$, $m^{\Oh(\td\cdot 2^{\td})}$.
That is, in particular, \Cref{thm:main-result} generalizes \Cref{thm:vertex-cover}.

For constant number of distinct arc costs, the algorithm of \Cref{thm:main-result} is more efficient with running time $p^{\Oh(\tw\cdot |W|)}\cdot m^{\Oh(1)}$.
In particular, \tpathedit is fixed-parameter tractable with respect to both $\fvs+|W|$ and $\td+|W|$.

\section{Hardness of \tpathedit}

In this section we prove hardness of \tpathedit, showing that our \classXP-algorithms cannot be turned into \classFPT-algorithms without breaking $\classFPT\neq \classW{1}$.

\begin{theorem} \label{W[1]-hardness}
\tpathedit is \classW{1}-hard with respect to:
\begin{itemize}
    \item $n$, when parallel arcs in $G$ are allowed,
    \item $p+ \vc$, when parallel arcs in $G$ are forbidden.
\end{itemize}
\end{theorem}

\begin{proof}
We give a parameterized gap reduction from \textsc{Modified $k$-Sum} (\textsc{M$k$S}), which is \classW{1}-hard \cite{AAAI22_DFI}.
In \textsc{M$k$S}, we are given $k$ integer sets $X_1, \ldots, X_k$ and the target sum $S$.
The goal is to determine whether there exists a choice of $k$ integers, one from each of $X_1, X_2,\ldots, X_k$, with a total sum $S$.

Given an instance of $(X_1, \ldots, X_k,S)$ of \textsc{M$k$S}, we construct a corresponding \tpathedit instance as illustrated in~\Cref{fig:Reduction}. For each set $X_i$, we create a gadget consisting of two vertices $v_i$ and $v_{i+1}$ connected by multiple parallel arcs, each representing an element $x \in X_i$ with weight $w = x$. The goal is to ensure that the agent reaches vertex $t$ via the unique $T$-path $sabt$ including the critical arc $(b, t)$.

The construction uses the following parameters:
     $\beta = \frac{1}{4} + \varepsilon$ for arbitrary $0 < \varepsilon < 1$;
     $r = 10 \cdot S$;
     $A = \emptyset$;
     $T = \{(b, t)\}$;
     $x = S + 2$.
We duplicate each of the arcs of $(s, c)$, $(c, t)$ into $z := |X_1| + \ldots + |X_k|$ identical copies with the same arc cost. Breaking the path $s\to c \to t$ in $G$  uses at least $z$ edits.

We will show that the answer to the \textsc{Modified $k$-Sum} problem is ``Yes'' if and only if there is a solution for the \tpathedit problem of size no greater than $z-k$.

Suppose that the \tpathedit instance has a feasible solution of size at most $z-k$. Then, the only arcs removed are those between vertices $v_1, \ldots, v_k, t$ in gadgets; removing any arc from the lower part of the graph would be too costly or would block the required edge $(b, t)$. Let $G'$ denote the modified graph, and let $S'$ be the cost of the shortest $v_1$-$t$ path in $G'$.

We now analyze possible values of $S'$. If $S' > x - 2$, then $S' \geq x - 1$, and the perceived cost of any path starting with $s\to a\to \ldots$ is at least $\beta \cdot (1 + x-1)$, which is worse than the perceived cost of the lower path $s\to c\to t$, that is equal to $\beta \cdot (x - \varepsilon)$. Hence, the agent would not take the $s\to a\to b\to t$ path.
If $S' < x - 2$, then $S' \leq x - 3$, and, when he makes decision in the vertex $a$, the agent compares
  \[
  1 + \beta \cdot S' < 0 + \beta \cdot (x + 1),
  \]
  and prefers to go through $v_1$ rather than directly to $b$.

Therefore, the only way for the agent to traverse the desired arc $(b, t)$ is when $S' = S = x - 2$. This implies the existence of a sequence of arc weights corresponding to a valid solution of the original \textsc{Modified $k$-Sum} instance. Conversely, any such solution induces a deletion pattern in \tpathedit with no more than $z-k$ edits, that forces the agent to pass through the arc $(b, t)$ as desired.

At the same time, as is evident from the proof, the only solution of another type for such \tpathedit instance---to break the path $s\to c \to t$---will have a size at least $z$. Which completes our reduction.

\begin{figure}[ht]
 \center{\includegraphics[scale=0.30]{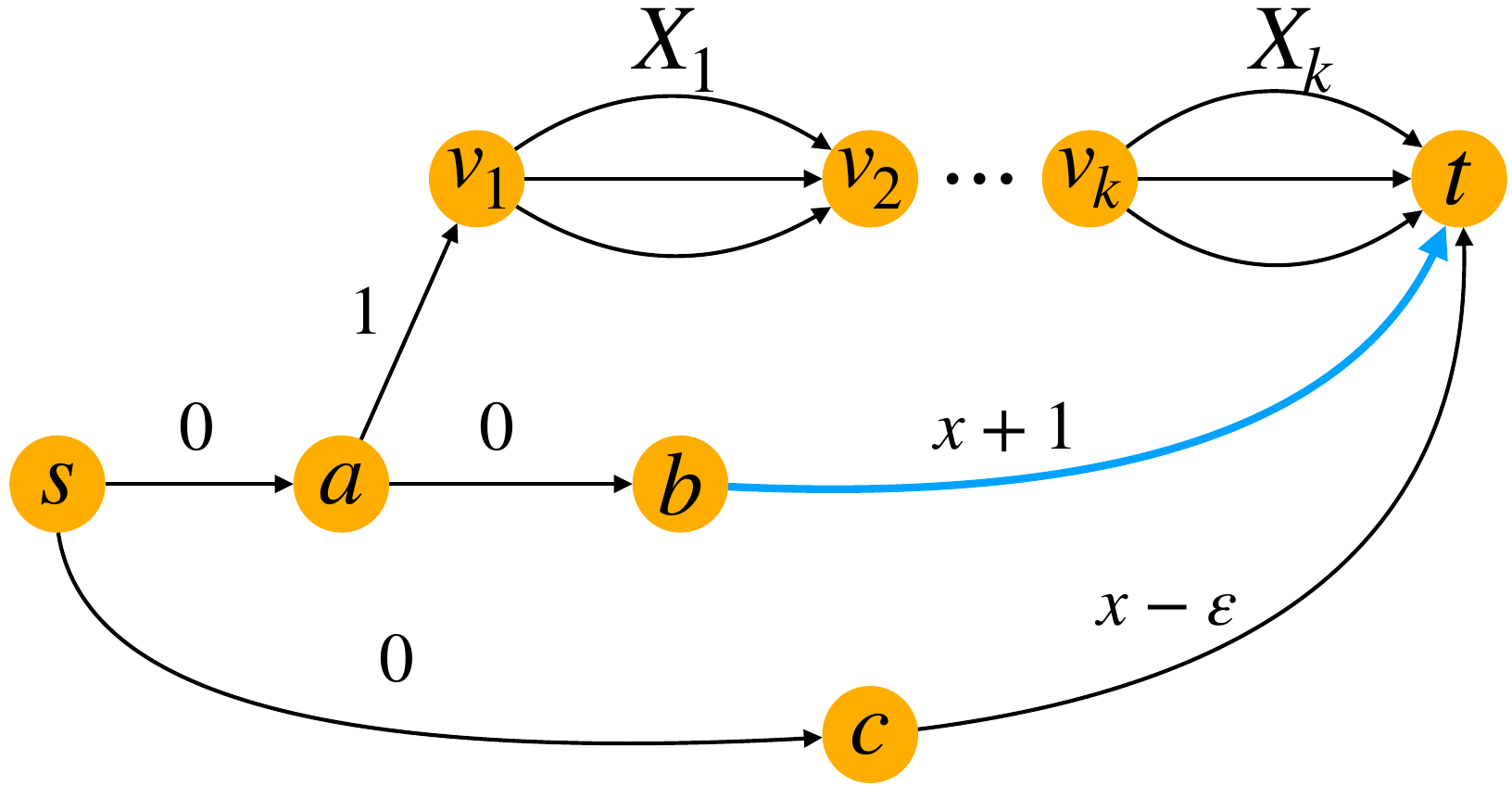}}
 \caption{The construction of graph $G$ from \Cref{W[1]-hardness}.}
 \label{fig:Reduction}
\end{figure}

Note that the construction of the graph $G$ involves only $k+5$  vertices. Int the same time, this example can be easily modified into an example without parallel arcs by adding an additional vertex in the middle of each arc. 
In such a graph, the vertex cover number will not exceed $k+5$, and the parameter $p$, the maximal arc length of a path in $G$, is equal to $2k+2$. Thus, we have obtained a parameterized reduction to \tpathedit with parameter $n = k+5$ in one case, and $p+\vc \leq 3k + 7$ in the other.

This completes the reduction, establishing \classW{1}-hardness for both settings.
\end{proof}

The graph in the reduction of \Cref{W[1]-hardness} is a series-parallel graph, and such graphs are known to have treewidth at most two \cite{graphclasses}.

\begin{corollary}
\tpathedit is \classNP-hard for $\tw=2$.
\end{corollary}

Note that this result is tight: when $\tw (G+A) = 1$, the graph is a tree, and the \tpathedit problem becomes trivial. In this case, the agent has a unique $s$–$t$ path, and the decision reduces to checking whether agent follow this path and whether it contains all arcs in $T$, with no modifications required.

On the other hand, \cite{ECAI24_BelovaDFGI24} proved that the \tpathdel problem---a special case of our \tpathedit formulation---remains \classNP-hard even when the maximum length $p$ of any $s$–$t$ path in $G$ is bounded by $8$. This result immediately implies that \tpathedit is \classParaNP-hard with respect to the parameter $p$.

\section{Conclusion}

We have introduced and studied the \tpathedit problem, a general framework for designing structural interventions that guide time-inconsistent agents toward completing critical tasks. Our model extends the Kleinberg–Oren framework by formalizing the principal’s role in influencing agent behavior through minimal graph edits.

Our results delineate the tractability frontier of \tpathedit across a range of structural parameters, including treewidth, vertex cover, feedback vertex set, and path length. In particular, we demonstrate fixed-parameter tractability for combined parameters capturing graph structure and cost diversity, while also establishing lower bounds via \classParaNP- and \classW{1}-hardness.

An intriguing open question is to determine whether \tpathedit admits an $f(\tw) \cdot m^{g(p)}$-time algorithm for computable functions $f$ and $g$.

\bibliography{planning}
\bibliographystyle{apalike}

\appendix

\section{Appendix. The Main Result}\label{appendix}
In this section, we prove \Cref{thm:main-result}.
We require a notion of nice tree decompositions.
For comprehensive introduction to tree decompositions and dynamic programming over tree decompositions, we refer the reader to the book \cite{CyganFKLMPPS15}.

\begin{definition}[Nice tree decomposition]
    A \emph{nice tree decomposition} of a graph $G = (V, E)$ is a special form of a tree decomposition $(\mathcal{T}, \{(B_a, E_a)\}_{a \in V(\mathcal{T})})$ where:
    \begin{itemize}
        \item The tree $\mathcal{T}$ is rooted, it means that we distinguish one vertex $a_0$ of $\mathcal{T}$ which is the root of $\mathcal{T}$. Its bag is empty: $B_{a_0} = \emptyset$; and its set of edges is complete: $E_{a_0} = E(G)$.
        Each node $a$ is of one of the following five types:
        \begin{itemize}
            \item Leaf node: has no children and its bag and edges are empty: $B_a = E_a = \emptyset$.
            \item Introduce vertex node: has one child $b$ and its bag adds one vertex: $B_a = B_{b} \cup {v}$ for some $v \notin B_{b}$, and $E_a = E_b$.
            \item Forget node: has one child $b$ and its bag removes one vertex: $B_a = B_{b} \setminus {v}$ for some $v \in B_{b}$, and $E_a = E_b$.
            \item Join node: has two children $b_1$ and $b_2$ such that $B_a = B_{b_1} = B_{b_2}$ and $E_a = E_{b_1} \cup E_{b_2}.$
            \item Introduce edge node: has one child $b$, $B_a = B_b$ and there is an edge $uv \in E \setminus E_b$ such that $E_a = E_b \cup \{uv\}$.
        \end{itemize}
        \item The bags $B_a$ still satisfy the standard tree decomposition properties:
        \begin{itemize}
        \item For every vertex $v \in V(G)$, the set of nodes $\{a \in V(\mathcal{T}) \mid v \in B_a\}$ forms a connected subtree of $\mathcal{T}$.
        \item For every edge $(u,v) \in E(G)$, there exists a bag $B_a$ such that $u,v \in B_a$.
        \end{itemize}
    \end{itemize}
\end{definition}

The treewidth of a graph is the minimum possible width of its nice tree decomposition.

\begin{definition}
    The \emph{width} of a nice tree decomposition $(\mathcal{T}, \{(B_a, E_a)\}_{A \in \mathcal{T}})$ is defined as
\[
\max_{a \in \mathcal{T}} |B_a| - 1,
\]
that is, one less than the size of the largest bag.
\end{definition}

While computing the treewidth of a graph is \classNP-hard, we are able to find nice tree decompositions of small width efficiently.

\begin{proposition}[\cite{CyganFKLMPPS15}]
    Nice tree decomposition of width $\Oh(\tw)$ consisting of at most ${m}^{\mathcal{O}(1)}$ nodes can be computed in $2^{\mathcal{O}(\tw)}\cdot m^{\Oh(1)}$ time.
\end{proposition}

To make the agent traverse all arcs of $T$, we will force it to not follow any $T$-avoiding arcs.

\begin{definition}[$T$-avoiding arc]
We say that an arc $e$ going from $u$ to $v$ in $G+A$ is \emph{$T$-avoiding}, if there is an arc $x\to y$ in $T$, but all of the following is true: 
\begin{enumerate}
    \item $\{u,v\}\neq \{x,y\}$;
    \item $x$ is reachable from $u$ in $G+A$.
    \item $x$ is not reachable from $v$ in $G+A$.
\end{enumerate}
\end{definition}

Clearly, if an $s$-$t$ path does not contain any $T$-avoiding arc, then it traverses all arcs of $T$.

We are almost ready to prove \Cref{thm:main-result}.
Before giving the remaining auxiliary results and the proof itself, we first restate \Cref{thm:main-result} here for convenience. 

\mainResult*

We note that our algorithm makes all necessary arc removals and additions in order to force the agent with \emph{arbitrary} arc priorities to traverse all arcs of $T$.
That is, if an agent has two arcs with two identical perceived costs at a vertex, the agent might choose any of them, and we should make sure that \emph{any} path that the agent could follow traverses all arcs of $T$.

We require the explicit construction of $L$ in our algorithm.
We show that $L$ can be computed in polynomial time.

\begin{lemma}
    Given an instance of \tpathedit, $L$ can be computed in $(|L|+m)^{\Oh(1)}$ time.
\end{lemma}
\begin{proof}
  We compute $L(v)$, equal to the set of all possible $v$-$t$ path costs for each $v\in V(G)$, by the following recursive relation:
  $$L(u)=\bigcup_{e\in E(G+A), e:u\to v} \{w(e)+\ell \mid \ell \in L(v)\},$$
  and $L(t)=\{0\}.$
  This can be computed by following a topological ordering of $G+A$, where $t$ is the first vertex, and $s$ is the last.

  For each vertex, the computation is done in at most $m\cdot |L|\cdot |L|$ running time.
  The total running time is polynomial in $m+|L|$.
\end{proof}

We are ready to prove the main result.

\begin{proof}[Proof of \Cref{thm:main-result}]
We describe an algorithm with the desired running time.

    We assume that we are given an instance of \tpathedit where $|L|>1$, otherwise $|L|=\{0\}$, that is, each arc cost is $0$, and the problem can be solved in  polynomial time.
    We also assume that every node is reachable from $s$ in $G+A$, and $t$ is reachable from every node in $G+A$.
    
    Then, the nice tree decomposition  $\mathcal{T}$ of $G+A$ of width $\Oh(\tw)$ is computed in $2^{\Oh(\tw)}\cdot {m^{\Oh(1)}}$ running time.
    We assume that the introduce edge node is present for each arc of $G+A$.
    That is, if there are $z$ parallel arcs between $u$ and $v$ in $G+A$, there are $z$ corresponding introduce edge nodes in $\mathcal{T}$.
    
    We additionally modify $\mathcal{T}$ by putting $s,t\in B_a$ into the bag of every node $a\in V(\mathcal{T})$, and the root node $a_0$ has its bag $B_{a_0}$ consisting of exactly $s$ and $t$, that is, $B_{a_0}=\{s,t\}$.
    We also force our nice tree decomposition to not introduce or forget $s$ or $t$.
    That is, the leaves of $\mathcal{T}$ also have $B_a=\{s,t\}$.
    The width of $\mathcal{T}$ remains bounded by $\Oh(\tw)$.

    Having our nice tree decomposition $\mathcal{T}$, we compute
    $$\operatorname{OPT}(a, D, d, R, F_0, F_1, f),$$
    where
    \begin{itemize}
        \item $a$ is a node of $\mathcal{T}$.

        \emph{The search space for arc edits for $\operatorname{OPT}(a,\ldots)$ is $E_a$.
        }
        \item $D$ is a mapping from $B_a$ to $L\cup \{+\infty\}$.
        
        \emph{Intuitively, $D(v)$ stores the distance from $v$ to $t$ in the solution graph.
        }
        \item $d$ is a mapping from $B_a$ to $\{0,1\}$.
        
        \emph{Intuitively, $d(v)=1$ if we are sure that the first arc of the shortest $v$-$t$ path belongs to the solution graph.}
        \item $R$ is a subset of $B_a$.
        
        \emph{The intuition behind $R$ (the set of reachable vertices) is that it consists of all vertices that the agent can reach in the solution graph.}
        \item $F_0, F_1$ are mappings from $R$ to $L$.

        \emph{These has the following meaning: if the agent is at $v\in R$ in the solution graph, then there is a $v$-$t$ path (with the first arc $v\to x$ for some $x\in V(G)$) in the solution graph that the agent can follow, with the first arc cost $F_0(v)-F_1(v)$, and distance from $x$ to $t$ equal to $F_1(v)$.}
        
        \item $f$ is a mapping from $R$ to $\{0,1\}$.

        \emph{This plays the similar role for $F_0, F_1$ that $d$ plays for $D$.
        Intuitively, $f(v)=1$ if we are sure that the arc outgoing of $v$ corresponding to $F_0(v), F_1(v)$, is present in the solution graph.}
    \end{itemize}
\newcommand{\rss}[0]{E'_a}

    Formally,
    $$\operatorname{OPT}(a, D, d, R, F_0, F_1, f)$$
    equals the minimum number of arc edits (or $+\infty$ if such sequence of edits does not exist) in the initial arc set $E_a\setminus A$ (where additions of arcs from $E_a\cap A$ is allowed), such that the resulting set $\rss\subset E_a$ satisfies the following eight conditions.
    Below we use $$\dist'(u,t)=\min_{v\in B_a\setminus \{u\}}\left\{\dist_{E'_a}(u,v)+D(v)\right\},$$
    if $d(u)=0$, and
    $$\dist'(u,t)=D(u),$$
    if $d(u)=1$.
    \begin{enumerate}
        \item\emph{(Distances are shortest).} 
        
        For each $u\in B_a$, $$D(u)\le\min_{v\in B_a\setminus \{u\}}\left\{\dist_{E'_a}(u,v)+D(v)\right\}.$$
        \item\emph{(Distances are real).} 
        
        For each $u \in B_a$, if $d(u)=0$, then 
        $$D(u)=\dist'(u,t).$$
        
        \item\emph{(Follow-arcs are followed by the agent).} 
        
        For each $u\in R$ with $u\neq t$, for each edge $e\in \rss$ outgoing of $u$ to some $v\in V(G)$,
        $$(F_0(u)-F_1(u))+\beta \cdot F_1(u)\le w(e)+\beta\cdot\dist'(v,t).$$

        \item\emph{($R$ is not abandoned by the agent).}
        
        For each $u \in R$ such that $u\neq t$, $$(F_0(u)-F_1(u))+\beta\cdot F_1(u)\le \beta\cdot r.$$

        \item \emph{(Follow-arcs are real).} 
        
        For each $u\in R$, such that $f(u)=0$, there exists $e\in 
    \rss$ going from $u$ to some $v\in V(G)$ such that
        $$F_1(u)=\dist'(v,t) \text{ and } F_0(u)=w(e)+F_1(u).$$
        
        \item\emph{($R$ is reached from $R$).} 
        
        For each $u\in R$, such that $f(u)=0$, there exists $v\in R$ and a $u$-$v$ path $P'\subset \rss$ with arcs $e_1, e_2, \ldots, e_p$ and vertices $u=u_0, u_1, \ldots, u_p=v$, such that for every $i \in [p]$,
        $$w(e_i)+\beta\cdot \dist'(u_i,t)\le w(e')+\beta\cdot\dist'(x,t),$$
        for every arc $e'\in\rss$ going from $u_{i-1}$ to some $x\in V(G)$, and
        $$w(e_i)+\beta\cdot\dist'(u_i,t)\le \beta\cdot r.$$

        \item\emph{(Any agent's path traverses $T$).} 
        
        For each $u, v\in R$ and for each $u$-$v$ path $P'\subset E'_a$ satisfying the two conditions of the previous point, for each $i\in[p]$, $e_i$ is not a $T$-avoiding arc. 

        \item\emph{(Agent can't follow an arc out of $R$).} 
        
        For each $u\in R$ and $v\notin R$, for each $u$-$v$ path $P'\subset \rss$, with arcs $e_1, e_2,\ldots, e_p$ and vertices $u=u_0, u_1, \ldots, u_p=v$, there exists $i\in[p]$ with
        $$w(e_i)+\beta\cdot \dist'(u_i,t)> w(e')+\beta\cdot\dist'(x,t),$$
         where $e'\in\rss$ is some arc going from $u_{i-1}$ to some $x\in V(G)$.
    \end{enumerate}
    
    If an optimal solution graph $G^*$ exists, and one of the agent's path in $G^*$ is $P$, that starts with $e_1$ going from $s$ to some $u_1\in V(G)$, then for $\ell^*=\dist_{G^*}(s,t)$, $\ell^*_1=\dist_{G^*}(u_1, t)$, $\ell^*_0=w(e_1)+\ell^*_1$ and
    \begin{itemize}
        \item $R^*=\{s,t\}$;
        \item $D^*=\{\ell^*,0\}$;
        \item $d^*=f^*=\{s\to 0, t\to 1\}$;
        \item $F^*_0=\{s\to \ell^*_0, t\to 0\};$
        \item $F^*_1=\{s\to \ell^*_1, t\to 0\};$
    \end{itemize}
    $\operatorname{OPT}(a_0, D^*,d^*,R^*,F_0^*,F_1^*,f^*)$ equals the number of arc edits required to obtain $G^*$.
    This follows just from the definition of $\operatorname{OPT}$.

    We move on to computing  $\operatorname{OPT}(a, D,d,R,F_0,F_1,f)$ (for arbitrary values of $a, D,d,R,F_0,F_1,f$) via dynamic programming over $\mathcal{T}$.
    Note that choices of parameters violating the condition 4 give the corresponding value of $\operatorname{OPT}$ equal to $+\infty$.
    We assume that condition 4 always holds for a choice of arguments of  $\operatorname{OPT}$.
    Similarly, we assume that $s,t\in R$, $D(t)=0$ and $D(s)<+\infty$ always hold as well.
    
    The computing is done in a recursive top-down approach on the structure of $\mathcal{T}$.
    There are several cases of processing a node via its children nodes, depending on the tree decomposition node type.
    
    \medskip\noindent\textbf{Forget node.}
    The node $a$ has exactly one child $b$.
    Then $E_a=E_b$ and $B_a=B_b\setminus\{u\}$ for some $u\in V(G)$.
    By definition, $\operatorname{OPT}(a, D,d,R,F_0,F_1,f)$ equals to the minimum of
    \begin{multline*}
    \min\limits_{\ell,\ell_0,\ell_1\in L}
    \operatorname{OPT}(b,\\D\cup\{u\to\ell\},d\cup\{u\to 0\},\\R\cup \{u\},\\F_0\cup\{u\to\ell_0\}, F_1\cup\{u\to\ell_1\}, \\f\cup\{u\to 0\}),
    \end{multline*}
    corresponding to the solution where $u$ is reachable, and
    $$\min\limits_{\ell\in L\cup\{+\infty\}}\operatorname{OPT}(b,D\cup \{u\to\ell\}, d\cup\{u\to 0\},R,F_0,F_1,f),$$
    which corresponds to the case where $u$ is not on any agent's path.

    Note that $\ell_0, \ell_1$ in the first part are only taken over those satisfying $(\ell_0-\ell_1)+\beta\cdot \ell_1 \le \beta\cdot r$ (condition 4).
    
    \medskip\noindent\textbf{Introduce vertex node.}
    The node $a$ has exactly one child $b$.
    Then $E_a=E_b$ and $B_a=B_b\cup\{u\}$ for some $u\in V(G)$.
    Moreover, there are no arcs incident to $u$ in $E_a$.
    By definition, $\operatorname{OPT}(a, D,d,R,F_0,F_1,f)$ equals $+\infty$, if $d(u)=0$, but $D(u)\neq +\infty$ because for any choice of $\rss$ and $v\in V(G)$, $\dist'(u,v)=+\infty$.
    Similarly, it equals $+\infty$ if $u\in R$ but $f(u)= 0$.
    Otherwise, it equals
    $$\operatorname{OPT}(a, D\setminus \{u\},d\setminus \{u\},R,F_0,F_1,f),$$
    if $u\notin R$, or
    $$\operatorname{OPT}(a, D\setminus \{u\},d\setminus \{u\},R\setminus \{u\},F_0\setminus \{u\},F_1\setminus \{u\},f\setminus \{u\}),$$
    if $u\in R$.

    \medskip\noindent\textbf{Introduce edge node.}
    The node $a$ has exactly one child $b$. 
    Then $B_a=B_b$ and $E_a=E_b\cup\{e\}$ for some arc $e\in E(G)\cup A$ connecting some pair of vertices $u,v\in V(G)$.
    It is either $e\in E(G)$, then we consider either keeping $e$ in the graph or deleting it, or $e\in A$, then we choose between adding or skipping it respectively.
    These cases are in fact identical, and the difference is only the edit cost: we pay $0$ for having $e\in E(G)$ in the solution, and we pay $1$ for not having $e\in E(G)$ in the solution.
    For $e\in A$, we pay $1$ and $0$ respectively instead.
    In what follows, by $c_1$ and $c_0$ we denote the costs of having $e$ or not having $e$ in the solution respectively. 

    If $e$ does not belong to the optimal solution for $\operatorname{OPT}(a, D,d,R,F_0,F_1,f)$, then it equals
    \begin{equation}\label{intro-edge-1}
    \operatorname{OPT}(b, D,d,R,F_0,F_1,f)+c_0,
    \end{equation}
    since the solution profile is the same for $b$ but we pay something for not having the edge in the solution.

    Other possible case is when $e$ belongs to the solution.
    This is only possible if:
    \begin{itemize}
        \item $D(u)\le w(e)+D(v),$ and,
        \item $(F_0(u)-F_1(u))+\beta \cdot F_1(u)\le w(e)+\beta\cdot D(v),$ and
    \end{itemize}
    if $u\in R$ but $v\notin R$, or $e$ is $T$-avoiding, the last inequality should be strict (the agent should not follow $e$).

    If $e$ is taken into the solution, it enforces some changes in the profile.
    If $e$ satisfies $D(u)=w(e)+D(v)$ (condition 2), then we can propagate $d(u)=1$ for the profile of $b$.
    If $e$ satisfies $F_0(u)-F_1(u)=w(e)$ and $F_1(v)=D(v)$ (condition 5), then we propagate $f(u)=1$.
    Denote the obtained mappings by $d'$ and $f'$.
    Then (if $e$ should be taken into the optimal solution for the current profile), $\operatorname{OPT}(a, D,d,R,F_0,F_1,f)$ equals
    \begin{equation}\label{intro-edge-2}
    \operatorname{OPT}(b, D,d',R,F_0,F_1,f')+c_1.
    \end{equation}

    Consequently, to evaluate $\operatorname{OPT}(a, D,d,R,F_0,F_1,f)$, we just take the minimum of \eqref{intro-edge-1} and \eqref{intro-edge-2}, where the second is taken into account only if having this arc is compatible with the profile as discussed in the previous paragraph.

    \medskip\noindent\textbf{Leaf node.}
    Then $a$ has no children in $\mathcal{T}$, $E_a=\emptyset$ and $B_a=\{s,t\}$.
    To evaluate $\operatorname{OPT}(a, D,d,R,F_0,F_1,f)$, we just have to check $d(s)=1$ and $f(s)=1$.
    We set the value to $0$ if both holds, and put $\operatorname{OPT}$ equal to $+\infty$ otherwise.
    This agrees with the definition of $\operatorname{OPT}$.

    \medskip\noindent\textbf{Join node.}
    This is the only case when $a$ has two child nodes $b_1,b_2$.
    Then $B_a=B_{b_1}=B_{b_2}$, and $E_{b_1}$ and $E_{b_2}$ form a partition of $E_a$.
    To compute $\operatorname{OPT}(a, D,d,R,F_0,F_1,f)$, we have to split $d$ and $f$ into two separate profiles and sum up the edits in $E_{b_1}$ and $E_{b_2}$.
    This is due the distances ($D$) or follow-arcs ($F_0$, $F_1$) can be implemented (conditions 2 and 5 in definition of $\operatorname{OPT}$) in either $E_{b_1}$ or $E_{b_2}$, and we have to consider every possible case.
    Then $\operatorname{OPT}(a, D,d,R,F_0,F_1,f)$ equals
    $$\min\limits_{
    \begin{matrix}
    d^1\cdot d^2=d\\
    f^1\cdot f^2=f
    \end{matrix}
    }\left[
    \sum_{j=1}^2
    \operatorname{OPT}(b_j, D,d^j,R,F_0,F_1,f^j) \right],$$
    where $d^1,d^2$ and $f^1,f^2$ are taken over all mappings from $B_a$ and $R$ respectively to $\{0,1\}$.
    The pairs should satisfy $d^1(u)\cdot d^2(u)=d(u)$ for each $u\in B_a$ and $f^1(u)\cdot f^2(u)=f(u)$ for each $u\in R$.

    The list of the cases of tree decomposition node types is exhausted.
    This finishes the description of the dynamic programming routine of the algorithm.

\medskip\noindent\textbf{Restoring the answer from dynamic programming values.}
To find the minimum number of arc edits, we iterate all $\ell^*,\ell^*_0,\ell^*_1\in L$ and construct $D^*,d^*,R^*,F^*_0,F^*_1,f^*$ as in the beginning of this proof.
The minimum number of arc edits equals 
$$\min\limits_{\ell^*,\ell^*_0,\ell^*_1\in L}\operatorname{OPT}(a_0,D^*,d^*,R^*,F^*_0,F^*_1,f^*).$$

The sequence of edits can be restored from the transitions of the dynamic programming in the usual way.

\medskip\noindent\textbf{Correctness.}
First, $\operatorname{OPT}(a_0,D^*,d^*,R^*,F^*_0,F^*_1,f^*)\le k$, where $k$ is the minimum number of arc edits required to obtain the solution graph $G^*$.
To see this, one can follow the top-down computation from the root state $\operatorname{OPT}(a_0, D^*,d^*,R^*,F^*_0,F^*_1,f^*)$ to the leaves, in each node making the choices corresponding to the solution graph $G^*$.

Second, to prove $\operatorname{OPT}(a_0,D,d,\{s,t\},F_0,F_1,f)\ge k$, one can restore the arc edit sequence $\operatorname{OPT}(a_0,D,d,\{s,t\},F_0,F_1,f)$ in the way natural for dynamic programming and show that it forms a solution to the initial instance of \tpathedit.
Indeed, the restored dynamic programming transition sequence guarantees that $R\subset V(G)$ forms a set of vertices reachable by the agent.
There is no arcs from $R$ to $V(G)\setminus R$ that the agent could follow by having the least possible perceived cost.
On the other hand, the agent has an arc from $v$ to $u\in R$ that it can follow at each vertex $v\in R\setminus \{t\}$.
Finally, there is no $T$-avoiding arc that the agent could follow.

\medskip\noindent\textbf{Running time.}
To count the number of states of $\OPT$, note that \begin{itemize}
    \item $a$ has $|V(\mathcal{T})|=m^{\Oh(1)}$ options;
    \item $R, d, f$ have $2^{\Oh(\tw)}$ options;
    \item $D$ has $(|L|+1)^{\Oh(\tw)}$ options;
    \item $F_0, F_1$ have $|L|^{\Oh(\tw)}$ options.
\end{itemize}
In total, there are at most $(2|L|+2)^{\Oh(\tw)}\cdot m^{\Oh(1)}$ states for $\OPT$, which is upper-bounded by $|L|^{\Oh(\tw)}$ since $|L|\ge 2$.

It is left to give an upper bound for the time required to compute each value of $\OPT$.
Note that for each node type, the computation is done in time polynomial in $m$, except for the join nodes.
To compute $\OPT(a,D,d,R,F_0,F_1,f)$ in a join node, we consider at most $3^{|R|}$ options of $(d^1,d^2)$ and at most $3^{|R|}$ of $(f^1,f^2)$. 
This is also upper-bounded by $|L|^{\Oh(\tw)}$ since $|L|\ge 2$.
Restoring the dynamic programming transitions is done in the same time.

The total running time is thus at most $|L|^{\Oh(\tw)}$.
The proof is complete.
\end{proof}

\end{document}